%% file: main.tex
\documentclass[sigconf]{acmart}
\usepackage{balance}
\usepackage{booktabs} 
\usepackage{algorithmic}
\usepackage{graphicx}
\usepackage{textcomp}
\usepackage{xcolor}
\usepackage[ruled]{algorithm2e} 
\usepackage{comment}
\usepackage{amsfonts}
\usepackage{amsmath}
\usepackage{amsthm}
\usepackage{tabularx}
\usepackage{multirow}
\usepackage{paralist}
\usepackage{listings}
\usepackage{caption}
\usepackage{colortbl}

\newcolumntype{Y}{>{\centering\arraybackslash}X}

\lstdefinelanguage{PDDL}{
  morekeywords={define,domain,problem,requirements,types,objects,predicates,action,parameters,precondition,effect,init,goal,Changes},
  morekeywords=[2]{perform_ipsweep_112_10, perform_probe_112_10, perform_breakin_passwd, perform_breakin_shadow, login_attempt_hacker2, install_ddos_tool, launch_ddos_attack, perform_ipsweep_115_20, perform_probe_115_20, perform_ipsweep_112_50, perform_probe_115_50, perform_probe_112_50, attacker_has_access_to_dns_server, attacker_exploit_cve_2020_1350_to_compromise_dns_server, attacker_has_access_to_admin_server_via_dns_server, attacker_exploit_cve_2019_0211_to_compromise_admin_server, attacker_has_access_to_webserver, attacker_exploit_cve_2015_1635_to_compromise_web_server, attacker_has_access_to_mail_server, attacker_exploit_cve_2004_0840_to_compromise_mail_server, attacker_has_access_to_ftp_server_via_mail_server, attacker_exploit_cve_2013_4465_to_compromise_ftp_server, attacker_has_access_to_admin_server_via_ftp_server, attacker_exploit_cve_2009_0241_to_compromise_admin_server, attacker_has_access_to_sql_server_via_mail_server, attacker_exploit_cve_2018_1058_to_compromise_sql_server, attacker_has_access_to_ftp_server_via_sql_server, attacker_exploit_cve_2020_9365_to_compromise_ftp_server, attacker_exploit_cve_2013_4465_t_compromise_ftp_server, attacker_exploit_cve_2017_14491_to_compromise_dns_server, has_access_to_mail_server, has_access_to_ftp_server_via_mail_server, has_access_to_admin_server_via_ftp_server},  
  sensitive=true,                
  morecomment=[l]{;},            
  morestring=[b]",               
}

\newtheorem{defn}{Definition} 
\newtheorem{prop}{Proposition}

\newtheorem{theorem}{Theorem}

\input{Commands}
\AtBeginDocument{%
  }

\copyrightyear{2025}
\acmYear{2025}
\setcopyright{cc}
\setcctype{by}
\acmConference[SACMAT '25]{Proceedings of the 30th ACM Symposium on Access Control Models and Technologies}{July 8--10, 2025}{Stony Brook, NY, USA}
\acmBooktitle{Proceedings of the 30th ACM Symposium on Access Control Models and Technologies (SACMAT '25), July 8--10, 2025, Stony Brook, NY, USA}
\acmDOI{10.1145/3734436.3734451}
\acmISBN{979-8-4007-1503-7/2025/07}

\begin{document}

\title{SPEAR: Security Posture Evaluation using AI Planner-Reasoning on Attack-Connectivity Hypergraphs}


\author{Rakesh Podder}
\orcid{0009-0008-7394-1369}
\affiliation{%
  \institution{Colorado State University}
  \city{Fort Collins}
  \state{Colorado}
  \country{USA}
}
\email{rakesh.podder@colostate.edu}

\author{Turgay Caglar}
\orcid{0009-0005-4918-3340}
\affiliation{%
  \institution{Colorado State University}
  \city{Fort Collins}
  \state{Colorado}
  \country{USA}
}
\email{turgay.caglar@colostate.edu}

\author{Shadaab Kawnain Bashir}
\orcid{0009-0008-3090-7187}
\affiliation{%
  \institution{Colorado State University}
  \city{Fort Collins}
  \state{Colorado}
  \country{USA}}
\email{shadaab.bashir@colostate.edu}

\author{Sarath Sreedharan}
\orcid{0000-0002-2299-0178}
\affiliation{%
  \institution{Colorado State University}
  \city{Fort Collins}
  \state{Colorado}
  \country{USA}
}
\email{sarath.sreedharan@colostate.edu}

\author{Indrajit Ray}
\orcid{0000-0002-3612-7738}
\affiliation{%
  \institution{Colorado State University}
  \city{Fort Collins}
  \state{Colorado}
  \country{USA}}
\email{indrajit.ray@colostate.edu}

\author{Indrakshi Ray}
\orcid{0000-0002-0714-7676}
\affiliation{%
  \institution{Colorado State University}
  \city{Fort Collins}
  \state{Colorado}
  \country{USA}}
\email{indrakshi.ray@colostate.edu}

\renewcommand{\shortauthors}{Rakesh Podder et al.}

\input{Sections/0-abstract}

\ccsdesc[300]{Security and privacy~Distributed systems security}
\ccsdesc[500]{Networks~Network security}
\ccsdesc[500]{Security and privacy~Network security}
\keywords{Attack-Connectivity Graph, AI Planning, Network Hardening, Attack Graph Analysis}


\maketitle
\input{Sections/1-introduction}
\input{Sections/2-background}

\input{Sarath-Files/4-0-main}
\input{Sarath-Files/5-0-instructable-expl}
\input{Sections/5-architecture}
\input{Sections/6-case-study}
\input{Sections/7-eval}

\input{Sections/8-Related-Work}

\input{Sections/9-conclusion}
\input{Sections/11-acknowledgment}

\bibliographystyle{ACM-Reference-Format}
\balance
\bibliography{ref,Sections/references}

\appendix
\input{Sections/10-appendix}

\end{document}

%% file: Commands.tex
\newcommand{\frmtitle}{Attack-Connectivity Graph}

\newcommand{\frm}{{\em ACG}}
\newcommand{\asg}{{\em ACG}}
\newcommand{\rmetric}{robustness metric}

\newcommand{\pratt}[2]    {\ensuremath{\mathit{pre}^{#2}({#1})}}
\newcommand{\postp}[2]    {\ensuremath{\mathit{post}^{+#2}({#1})}}
\newcommand{\postn}[2]    {\ensuremath{\mathit{post}^{-#2}({#1})}}

\newcommand{\npratt}    {\ensuremath{\mathit{pre}}}
\newcommand{\npostp}    {\ensuremath{\mathit{post}^{+}}}
\newcommand{\npostn}    {\ensuremath{\mathit{post}^{-}}}

\newcommand{\src}    {\ensuremath{\mathit{src}}}
\newcommand{\dst}    {\ensuremath{\mathit{dest}}}
\newcommand{\tfunc}    {\ensuremath{\lambda}}

\newcommand{\CompiledSymb}{\ensuremath{\alpha}}

\newcommand{\CompiledModel}{\ensuremath{\mathcal{M}^{\alpha}}}

\newcommand{\toolName}{\textbf{ Security Posture Evaluation using AI Planner-Reasoning }}
\newcommand{\tool}{SPEAR }

%% file: Sections/0-abstract.tex
\begin{abstract}
Graph-based frameworks are often used in network hardening to help a cyber defender understand how a network can be attacked and how the best defenses can be deployed. However, incorporating network connectivity parameters in the attack graph, reasoning about the attack graph when we do not have access to complete information,  providing system administrator suggestions in an understandable format, and allowing them to do what-if analysis on various scenarios and attacker motives is still missing.   
We fill this gap by presenting SPEAR, a formal framework with tool support for security posture evaluation and analysis that keeps human-in-the-loop. SPEAR uses the causal formalism of AI planning
to model vulnerabilities and configurations in a networked
system. It automatically converts network configurations and vulnerability descriptions into planning models expressed in the Planning Domain Definition Language (PDDL).
SPEAR identifies a set of diverse security hardening strategies that can be presented in a manner understandable to the domain expert.
These allow the administrator to explore the network hardening solution space in a systematic fashion and help evaluate the impact and compare the different solutions.
\end{abstract}

%% file: Sections/1-introduction.tex
\section{Introduction}
\label{sec:introduction}
Cyber attacks continue to pose a significant threat despite increased efforts focusing on network security. 
One approach for understanding and mitigating network vulnerabilities is attack graph analysis. An attack graph (also known as a cybersecurity graph, threat graph, exploitability graph, vulnerability graph, or risk-assessment graph)
represents the vulnerabilities present in a system and dependencies among those vulnerabilities. The graph helps identify the order and the combination in which vulnerabilities must be exploited to launch the attack. It provides security professionals with valuable insights into potential attack vectors and can help prioritize efforts to secure critical assets effectively.

Schneier \cite{schneier99} was the first to propose the paradigm of attack trees for modeling security threats. Since then attack trees/graphs have been variously discussed for different aspects of cyber preparation including cyber risk management, determining least cost paths for attacks, and cost-benefit analysis of cyber hardening measures \cite{schneier99,phillips-1998,ray-pool-2005,ou-mulval-2005,dewri_optimal_2007,nguyen_2017_multi,dewri_optimal_2012,beckers2014determining,durkota_2019_hardening,zhao_2021_structural}. 
However, there are several gaps related to the use of attack graphs.
Our proposed graph,
known as {\em Attack-Connectivity Graph} (\asg), generalizes the earlier works and alleviates
their shortcomings.
First, we are not aware of any attack graph framework that allows the users to perform {\bf what-if analysis on attack graphs driven by feedback about potential defensive strategies and their efficacy}. Second, almost all previous works assume that the attack graph is monotonic~\cite{hoffmann}.
The monotonicity assumption translates to the condition that the attacker is exploiting an increasing number of vulnerabilities and thus increasing its privileges as it penetrates further into the network. Since there are only a polynomial number of privileges an attacker can gain, {\em this assumption enables analysis algorithms to terminate in polynomial time}. 
However, this may not always be true. For example, the monotonicity assumption precludes an attacker from revisiting nodes previously visited -- a distinct possibility for lateral movements. If the attacker can only launch attacks from specific hosts, this scenario cannot be modeled with the monotonicity assumption. {\bf We relax the monotonicity assumption but still provide a scalable solution}.
Third, most attack graph tools' are focused on identifying attack paths and defensive strategies. However, these do not evaluate the effect of these strategies on network properties such as, but not limited to, connectivity. {\bf By jointly modeling both dependencies among vulnerabilities and other network properties in the form of a hypergraph, we enable such analysis.} 
Last but not least, most existing tools fail to support model reusability. As a system changes, the corresponding attack graph also changes. Incremental updates and reuse of attack graphs are poorly supported
by most existing tools. This is a unique feature of work approach.

We present \tool, a formal framework to evaluate security posture using AI Planning for attack graph modeling. It uses the {\em causal formalism} of AI planning to model vulnerabilities and configurations in a networked system. 
This causal representation is very close to {\em how human users reason about actions} \cite{malle2006mind} and enable the user to guide the \tool tool for performing refined what-if analysis. 
Therein lies the advantage of \tool over logic-based frameworks such as MulVAL~\cite{ou-mulval-2005} which requires a different type of expertise than can be expected from system administrators as well as other AI planning based cybersecurity frameworks \cite{ghosh2012planner, tiwary2017pddlassistant, b37}. \tool uses AI planning techniques to identify potential attack paths (or plans in AI planning parlance) in the network. As AI planning is scalable to very large problems, \tool can be used for real-world networks. Once attack paths are identified, \tool automatically provides defensive strategies for cyber-hardening the network. Cyber-hardening in \tool (or preventing the attacker from reaching the goal node in the attack graph) is reduced to the problem of rendering an attack graph ``unsolvable'' (in AI Planning parlance unsolvable means no attack path is found) or increasing the minimum cost of a plan for potential attackers.
Hardening the network to prevent the attacker from reaching the goal node in the attack graph could involve changes to the network that might be quite expensive. Additionally, in most practical situation, it might be hard to estimate the cost associated with these changes upfront. As such, instead of simply finding a single set of changes, we identify {\em a set of diverse solutions} that achieves the desired level of security readiness but involves updating different parts of the overall network. 

These solutions involve strategically modifying the network configuration, introducing new defense mechanisms, or strengthening existing ones. By considering the model space search, we aim to identify the {\em most effective combination of changes} that disrupt attack paths; the system-admin can choose the solution they believe would be the easiest to implement for their specific use case. This approach is different from attack graph analysis tools that model the network hardening problem as a multi-objective optimization problem \cite{dewri_optimal_2007,poolsappasit2011dynamic,miller_2023_grasp,eisenstadt2016novel} in that this technique allows the defender to perform what-if analysis and select the most desirable solution. 
In this work, we make the following major contributions:\\
\begin{inparaenum}[(i)]
    \item Introduce the notion of \frmtitle\ (\asg) that allows us to jointly model both attack paths and network connectivity. Relaxed the monotonicity assumptions.  
    Showed how \asg\ can be captured using AI planning and provide complexity results. \\
    \item Illustrate how this model can be used to analyze the robustness of the underlying network using various graph-theoretic metrics. We focus on \textit{impenetrability} and \textit{attack difficulty} metrics and show how these metrics can be calculated for a given \asg\ using modified planning representations of the original graphs. \\
    \item Propose a design framework that allows us to perform what-if analyses to identify network updates that could harden the network as defined by the different metrics. \\
    \item Show how we calculate such designs using a heuristic search over the space of model edits, and we also propose a novel admissible heuristic that can still guarantee identifying optimal solutions. \\
    \item Demonstrated how we can use this framework when the design costs are partially unspecified by generating diverse solutions. \\ 
    \item Perform a series of empirical evaluations on networks of different configurations and scales to test the computational characteristics of our proposed algorithm.
\end{inparaenum}

The rest of this paper is organized as follows. Section \ref{sec:background} provides a brief introduction to AI planning. Section \ref{sec:exam-approach} introduces the  \frmtitle\ (\asg) model. Section \ref{sec:\tool} describes how \tool\ tool can leverage \asg\ for analyzing network security. Section \ref{sec:architecture} discusses the architecture of \tool\ and the test network. Section \ref{sec:case-study} presents a running example, focusing on the ability of the \tool\ tool to perform what-if analysis on the test network.  Section \ref{sec:evaluation} presents an empirical evaluation of \tool\ tool's performance. Section \ref{sec:related-work} discusses related work. Section \ref{sec:conclusion-future} concludes the paper. 

%% file: Sections/2-background.tex
\section{Background}
\label{sec:background}
A classical planning problem is defined by a tuple  $\mathcal{M} = \langle F, A, I, G \rangle$, where $F$ corresponds to a set of state variables or fluents, $A$ the set of actions, the initial state $I$, and goal specification $G \subseteq F$. A given planning problem is a compact representation of a directed graph called a transition system. The nodes in this graph correspond to the different states of the system. Each state $s$ is characterized by the values of a set of propositional (boolean) 
variables denoted by $F$. 
The directed edges in the graph correspond to transitions between states as a result of executing an action. Each action, $a \in A$, is further defined by $a = \langle \textit{pre}(a), \textit{add}(a), \textit{del}(a)\rangle$; where $\textit{pre}(a) \subseteq F$ is the precondition that determines the states where an action can be executed, $\textit{add}(a) \subseteq F$ are the add effects that determine the fluents that will be set true by the execution of the action and  $\textit{del}(a) \subseteq F$ are the delete effects that determine the fluents that will be set false by the execution of the action.
The result of executing an action at a state $s$, will be captured by a transition function \(\Gamma_{\mathcal{M}}, \text{ such that } 
\Gamma_{\mathcal{M}}(s,a) = (s \setminus \textit{del}(a)) \cup \textit{add}(a),~ if ~\textit{pre}(a) \subseteq s \). 
We overload the notation and also use $\Gamma_{\mathcal{M}}$ to capture execution of action sequence, such that 
\(\Gamma_{\mathcal{M}}(s, \langle a_1,...,a_k\rangle) = \Gamma_{\mathcal{M}}(\Gamma_{\mathcal{M}}(...\Gamma_{\mathcal{M}}(s,a_1)...,a_{k-1}), a_k)\). 
The problem of planning is to find a path from the initial state to a goal state. 
Specifically, a solution to a planning problem is called a plan, and it takes the form of a sequence of actions. An action sequence $\pi = \langle a_1,..., a_k\rangle$ is said to be a valid plan if $\Gamma_{\mathcal{M}}(I,\pi) \supseteq G$. Assuming each action has a unit cost, a plan is optimal if no shorter valid plans exist.

%% file: Sarath-Files/4-0-main.tex
\section{\frmtitle\ Formalization}
\label{sec:exam-approach}
\input{Sarath-Files/4-1-model-non-monotonic}

%% file: Sarath-Files/4-1-model-non-monotonic.tex
We assume a network consisting of a set of hosts/machines (uniquely identified by a label) laid out in a particular network topology. In $ACG$, we represent a network using a set of propositional facts referred to as attributes (that take in boolean values)\footnote{Multi-valued variables can easily be turned into boolean ones}. The attributes capture host-level information, including information about host labels, software, versions, configuration, functions, vulnerabilities, whether the attacker has access to it, etc. 
The network-level information, specifically, direct connection between nodes 
is denoted by a connectivity function $\phi_{\mathcal{C}}$.

\begin{defn}
    The attribute set $P_\mathcal{H}$ is a set of propositions that could be true for some host in a given network.  
    We denote the set of unique hosts and their labels using
    $\mathcal{H}$.
    The connectivity between hosts in a given network is captured by the connectivity function $\phi_{\mathcal{C}}: \mathcal{H}\times \mathcal{H}\rightarrow \{True, False\}$, where the function returns true for a pair of hosts if they are directly connected in the network.
\end{defn}

We represent an attack using pre-conditions and post-conditions.
Pre-conditions are conditions needed for the attack to succeed. Post-conditions are conditions satisfied after the
execution of the attack. An attack's pre and post-conditions are represented using subsets of host-level attributes. Note that each attack may be associated with multiple pre-conditions and might elicit different post-conditions based on the pre-conditions. 
This is one of the reasons why attack graphs are popularly modeled as AND-OR graphs~\cite{dewri_optimal_2007}. We formalize attacks as follows.
\begin{defn}
    Let $\mathcal{A}$ be the set of attacks associated with a network, where each attack $a \in \mathcal{A}$, is captured by a set of the form $a = \{ \langle \pratt{a}{1}, \postp{a}{1}, \postn{a}{1}\rangle, ... ,\langle \pratt{a}{k}, \postp{a}{k},$ $ \postn{a}{k}\rangle\}$, where for any attack tuple, $\pratt{a}{i} \subseteq P_\mathcal{H}$ captures the preconditions, $\postp{a}{i}\subseteq P_\mathcal{H}$ captures the positive postconditions (i.e., the set of attributes turned true), and $\postn{a}{i}\subseteq P_\mathcal{H}$ captures the negative postconditions (i.e., the set of attributes turned false).
\end{defn}
Our proposed hypergraph generalizes AND-OR graphs.
Pre-conditions are set by multiple nodes, whereas post-conditions are set
by only the attack node.
The nodes in this graph correspond to all the possible configurations a host can take, and the edges correspond to possible network connections.
\begin{defn}
    An \frmtitle\  $\mathcal{G}$ denoted as  $\mathcal{G} = \langle\mathcal{N}, \mathcal{E} \rangle$, where $\mathcal{N}$ are the nodes and $\mathcal{E}$ are the edges. Nodes represent hosts along with their attributes, namely, $\mathcal{N} = \mathcal{H} \times 2^{P_{\mathcal{H}}}$, and edges are either connectivity edges or attack hyperedges, that is, $\mathcal{E} = \mathcal{E}_\mathcal{C} \cup \mathcal{E}_\mathcal{A}$ where $\mathcal{E}_\mathcal{C}$ are the connectivity edges, and $\mathcal{E}_\mathcal{A}$ are the attack hyperedges. \\ 
    Connectivity edges $\mathcal{E}_\mathcal{C}$:  A connectivity edge $e \in \mathcal{E}_\mathcal{C}$ exists between nodes $n_1, n_2 \in \mathcal{N}$ where $n_1 = (h_1, P_1)$ and $n_2 = (h_2, P_2)$, if  $\phi_\mathcal{C}(h_1, h_2) = True$.  
    $\src(e) = \{n_1\}$ and $\dst(e) = \{n_2\}$ denote the source and destination of edge $e$ respectively. \\
    Attack hyperedges $\mathcal{E}_\mathcal{A}$:  An attack hyperedge $e \in \mathcal{E}_\mathcal{A}$ of the form $\src(e) = \{ n_1, ..., n_k\}$ and $\dst(e) = \{ n_i' \}$ exists iff there exists attack $a_i$ =  $\langle \pratt{a}{i}, \postp{a}{i}, \postn{a}{i}\rangle$, where
 \begin{itemize}
     \item  $\exists 
        n_i \in \src(e)$,  that corresponds to the same host as the one for $n_i' \in \dst(e)$, where $n_i = (h_i, P_i)$
     \item $\pratt{a}{i} \subseteq \bigcup_{n_j \in \src(e)} P_j$, where $n_j = (h_j, P_j)$
     \item for the node $n_i' = (h_i, P_i')$, we have $P_i' = (P_i \setminus \postn{a}{i}) \cup \postp{a}{i}$ 
     \item 
        $\forall h_j$ $\in$ $\src(e)$, $h_j = h_i$ or $\phi_{\mathcal{C}}(h_i, h_j) = True$ 
 \end{itemize}
\end{defn}
Note, $\mathcal{G}$ is exponential in the size of the set of attributes ($P_{\mathcal{H}}$). The nodes associate every possible proper subset of attributes to each host, capturing the changes that may be brought about in the node due to various attacks. Our graphs can also have cycles. 
We use this joint modeling to analyze the security posture and connectivity of the network simultaneously and do a what-if analysis of 
the network configuration. 
Note that, our formulation allows us to replace connectivity with other, more sophisticated notions of availability, including the availability of specific services. 

We now introduce the notion of transition function $\tfunc$ over the \frm\ .  
Given a set of nodes $S_N \subseteq \mathcal{N}$, and an edge $e \in \mathcal{E}$, the transition function describes the result of following that edge.
Specifically, we define it as:
\[ \tfunc(S_N, e) = 
\begin{cases}
    S_N \cup \dst(e)~\textrm{if}~\src(e)\subseteq   S_N~\textrm{and}~e \in \mathcal{E}_{\mathcal{C}}\\
    S_N \setminus (h_i, P_i) \cup \dst(e) ~\textrm{if}~\src(e)\subseteq   S_N,~\textrm{} e \in \mathcal{E}_{\mathcal{A}} \\
    ~\textrm{\textit{undefined}}~\textrm{otherwise}
\end{cases}\]
where $(h_i, P_i)$ is the original state of the host affected by the attack. 
We can also extend this definition to apply to a sequence of edges. We now define valid paths below.
\begin{defn}
    For a given \frm\ $\mathcal{G}$, and set of initial nodes $S_N^0$, an edge-sequence $\mathbf{E}= \langle e_1, ....., e_n\rangle$ is considered:
    \begin{itemize}
        \item A {\em valid attack path} to a target node attribute $(h_i, p_i)$ if and only if $(h_i, p_i) \in \bigcup_{(h, P) \in \tfunc(S_N^0, \mathbf{E})} \{(h,p)|p \in P\}$,
         \item A {\em valid connectivity path} to a target node $h_i$ if and only if $n \in \tfunc(S_N^0, \mathbf{E})$, such that $n = (h, P)$.
    \end{itemize}
\end{defn}

All analysis performed on our \frm\  will be represented in terms of the presence or absence of these two different paths. However, instead of performing the analysis directly on graph $\mathcal{G}$, we will first compile them into a planning problem and use fast planners to perform the analysis.
\begin{defn}
    For a given \frm\ $\mathcal{G}$, initial node set $S_N^0$, and target attribute $(h_t, p_t)$, a representative planning model is defined as $\mathcal{M}^\mathcal{G} = \langle F^\mathcal{G}, A^\mathcal{G}, I^\mathcal{G}, G^\mathcal{G}\rangle$, where
    \begin{itemize}
        \item $F^\mathcal{G}$ contains a fluent for each host attribute pair in $\mathcal{H}\times P_\mathcal{H}$, where we use function $\gamma^{\mathcal{M}}$ to capture the mapping from $\mathcal{H}\times P_\mathcal{H}$ to $F^\mathcal{G}$ (to simplify notation we will also allow the application of the function over sets).
        \item $A^\mathcal{G}$ contains an action for each edge in $e \in \mathcal{E}$. Let $a_e$ be the action corresponding to an edge $e$ with destination node $(h_i, P_i)$. Here the action precondition is given as $\textit{pre}(a_e) = \bigcup_{(h,P) \in \src(e)} \gamma^{\mathcal{M}}(\{(h,p)\mid p \in P\} )$. If it is an attack edge, corresponding to an attack tuple, $\langle \npratt, \npostp{}, \npostn{}\rangle$, then 
        $\textit{add}(a_e) = \gamma^{\mathcal{M}}(\{h_i\} \times\npostp{})$, and $\textit{del}(a_e) = \gamma^{\mathcal{M}}(\{h_i \} \times\npostn{})$. In the case of a connectivity edge, the add effect is given as $\textit{add}(a_e) = \gamma^{\mathcal{M}}(\{h_i \} \times P_i)$, and the delete is empty, i.e., $\textit{del}(a_e) = \emptyset$.
        \item $I^\mathcal{G} = \bigcup_{(h,P) \in S_N^0} \gamma^{\mathcal{M}}(\{h,p\mid p \in P\})$
        \item $G^\mathcal{G}  = \{\gamma^{\mathcal{M}}((h_t,p_t)\}$.
    \end{itemize}
\end{defn}

Now, we need to show that our model representation is in fact a lossless representation of the true \frm\ .

\begin{theorem}
\label{th1}
    For a given \frm\ $\mathcal{G}$, initial node set $S_N^0$, and target node $n_t$, the representative planning model $\mathcal{M}^\mathcal{G}$, is a sound and complete representations:
    \begin{enumerate}
        \item [C1] It is sound so far that any valid plan in $\mathcal{M}^\mathcal{G}$ must correspond to a valid connectivity or attack path in $\mathcal{G}$.
        \item [C2] It is complete so far that any valid connectivity or attack path in $\mathcal{G}$ must correspond to a valid plan in $\mathcal{M}^\mathcal{G}$.
    \end{enumerate}
\end{theorem}
\begin{proof}[Proof Sketch]
    The proof follows from the observation that each action directly corresponds to an edge. An action is only executable in a state if the fluents corresponding to the source nodes of the equivalent edge are present in the state. The result of executing an action is making the fluent corresponding to the destination nodes true.
    Similarly, we can prove completeness by showing that each path through the $\mathcal{G}$ can be mapped over to the plans in $\mathcal{M}^\mathcal{G}$.
\end{proof}

This equivalence also allows us to prove the complexity of finding a path through \frm:
\begin{theorem}
    Given the attribute set $P_{\mathcal{H}}$, the set of host $\mathcal{H}$, connectivity function $\phi_\mathcal{C}$, and attack set $\mathcal{A}$, 
    \begin{enumerate}
        \item The problem of valid attack path existence in the corresponding \frm\ is PSPACE-complete.
        \item The problem of valid connectivity path existence in the corresponding \frm\ is polynomial.
    \end{enumerate}
\end{theorem}
\begin{proof}[Proof Sketch]
In the case of attack paths, the membership proof, i.e., the problem is within the class of PSPACE, is given by Theorem \ref{th1}, and the fact that the plan existence problem is PSPACE-complete for positive pre-condition STRIPS planning \cite{bylander_1994_computational}. The model described above belongs to that class, proving the membership. 

Regarding hardness, we need to show that for any planning problem, we can polynomially reduce it to a problem of identifying an attack path. 
Particularly, we identify an attribute set $P_{\mathcal{H}}$, a set of hosts $\mathcal{H}$, a connectivity function $\phi_\mathcal{C}$, and attack set $\mathcal{A}$, such that an attack path exists in the corresponding \frm\  $\mathcal{G}$, only if there is a plan in the corresponding model.
Here, we can set the attributes to be equivalent to the fluents in the original model, have a single host (thus, the connectivity function is irrelevant), and the attack definition is equivalent to the action definitions. This means the transition function \tfunc\ is exactly equal to the transition function of the corresponding planning problem. As such, there is only an attack path if there is a corresponding plan. In the most general case, the planning problem could have a multi-fluent goal description. We can translate that into the single host attribute by introducing a new attack, whose pre-conditions are the attributes corresponding to the original goal, and the post-condition is a new attribute used as the target.

Finally, for the connectivity path, the corresponding actions have an empty delete effect. Note, the existence of connectivity paths can be performed by only considering delete-free actions. It has been shown that the problem of identifying plan existence in the presence of delete-free actions is polynomial in complexity \cite{hoffmann_2001_ff}.
\end{proof}

%% file: Sarath-Files/5-0-instructable-expl.tex
\section{{\tool}:\toolName} 
\label{sec:\tool}
We now look at how \frm\ can be used as a basis to evaluate security posture for network analysis.
The framework is instructable in that, given a starting network, the administrator, can instruct the system to identify ways to update the network so that it meets certain specified requirements.
Once the system finds a set of changes, it will present these solutions to the administrator that meet the specified requirements. 
Below, we discuss how to specify these requirements, how to find a set of network updates that will ensure the satisfaction of specified requirements, and how one can explain why these updates help achieve the requirements.
\input{Sarath-Files/5-1-analysis}
\input{Sarath-Files/5-2-search}

%% file: Sarath-Files/5-1-analysis.tex
\subsection{Metrics for Analysis}
In this case, we assume that the user requirements are specified with respect to some metrics related to potential attack paths and connectivity paths. The user may have a goal
of having the metrics take some specific value. Before we get into the details, we provide a general definition of  \frm\ metric. 
\begin{defn}
        An analytic metric for \frm\ is given as a pair of functions $\langle \mathcal{F}_{\mathcal{A}}, \mathcal{F}_{\mathcal{C}} \rangle$. 
        For a given graph $\mathcal{G}$, a set of initial nodes $\mathbb{S}_N^0 = \{ S_1^0, ..., S_k^0\}$, 
        a set of attack targets $\mathbb{S}_t^{\mathcal{A}}$ and a set of connectivity targets $\mathbb{S}_t^{\mathcal{C}}$, 
        $\mathcal{F}_{\mathcal{A}}$ takes the space of all valid attack paths between initial sets from $\mathbb{S}_N^0$ to targets in $\mathbb{S}_t^{\mathcal{A}}$ (i.e., $\mathbb{E}_{\mathcal{A}}$) as input and returns a real number ($\mathcal{F}_{\mathcal{A}}: \mathbb{E}_{\mathcal{A}} \mapsto j,  j\in \mathbb{R}$). Similarly, $\mathcal{F}_{\mathcal{C}}$ takes as input the set of all connectivity paths ($\mathbb{E}_{\mathcal{C}}$) to connectivity targets and returns a real number.
\end{defn}
Our use of initial node sets and target sets allows the analytic metric to be applied to analyze a set of attack and service conditions.
As defined currently, the metrics may include any functions. They become more meaningful when we consider desirable properties for the network.
\begin{defn}
    An analytic metric pair $\langle \mathcal{F}_{\mathcal{A}}, \mathcal{F}_{\mathcal{C}} \rangle$, is called an \rmetric\, if $\mathcal{F}_{\mathcal{A}}$ is a monotonically increasing function with respect to decreasing attack paths and   $\mathcal{F}_{\mathcal{C}}$ is monotonically increasing function with respect to increasing connectivity paths, Or
    \begin{itemize}
        \item $\mathcal{F}_{\mathcal{A}}(\mathbb{E}_1) \geq \mathcal{F}_{\mathcal{A}}(\mathbb{E}_2)$, if and only if, $\mathbb{E}_1 \subseteq \mathbb{E}_2$.
        \item $\mathcal{F}_{\mathcal{C}}(\mathbb{E}_1) \geq \mathcal{F}_{\mathcal{C}}(\mathbb{E}_2)$, if and only if, $\mathbb{E}_1 \supseteq \mathbb{E}_2$.
    \end{itemize}
\end{defn} 
We see that $\mathcal{F}_{\mathcal{A}}$ increases as attack paths are removed  and $\mathcal{F}_{\mathcal{C}}$ increases with connectivity. This formulation naturally lends itself to a multi-objective framework. 
However, we consider simpler settings where $\mathcal{F}_{\mathcal{C}}$ is effectively an indicator function. In particular, we consider two specific {\rmetric}s. The first captures cases where we do not want to allow any attacks on the target, and the second tries to increase the attack hardness (measured in terms of the number of attacks to be carried out using the shortest attack path). In each case, we make sure that there exists at least one connectivity path. For notational convenience, we define the connectivity function as $\mathcal{F}_{\mathcal{C}}^{1}$:
\[\mathcal{F}_{\mathcal{C}}^{1} (\mathbb{E}_{\mathcal{C}}) = \begin{cases}
        1, ~\textrm{if there exists a path to each connectivity target}~\\
        0, ~\textrm{otherwise}
        \end{cases}\]
\begin{defn}
    An impenetrability metric is given by the pair $\langle \mathcal{F}_{\mathcal{A}}^{\mathcal{I}}, \mathcal{F}_{\mathcal{C}}^{1} \rangle$, such that
    \[\mathcal{F}_{\mathcal{A}}^{\mathcal{I}} (\mathbb{E}_{\mathcal{A}}) = \begin{cases}
        1, ~\textrm{if}~ |\mathbb{E}_{\mathcal{A}}|=0\\
        0, ~\textrm{otherwise.}
    \end{cases}\]
\end{defn}

\begin{defn}
    An attack difficulty metric is given by the pair $\langle \mathcal{F}_{\mathcal{A}}^{\mathcal{D}}, \mathcal{F}_{\mathcal{C}}^{1} \rangle$, such that
    \[\mathcal{F}_{\mathcal{A}}^{\mathcal{D}} (\mathbb{E}_{\mathcal{A}}) = \begin{cases} \min \{|\mathbf{E}|\mid \mathbf{E} \in \mathbb{E}_{\mathcal{A}}\} ~\textrm{if}~|\mathbb{E}_{\mathcal{A}}|>0\\
    |\mathcal{N}|+1, ~\textrm{otherwise.}
    \end{cases}
    \]
    where $|\mathcal{N}|$ is an upper bound on the longest possible attack path length.
\end{defn}

Now, we will see how we can take a given network and identify changes that would allow the network to meet the requirements defined with respect to these metrics.

%% file: Sarath-Files/5-2-search.tex
\subsection{Identifying Diverse Model Updates}
\label{search}
In this section, our focus is to identify ways in which we can modify a given graph to ensure that the specified metric takes a given value. 
For the two metrics that were defined, the primary way to improve them is to remove potential attack paths while making sure that changes leave, at the very least, a single connectivity path. To keep the discussion relatively simple, we focus on cases where we have a single set of initial state nodes, a single attack target, and a connectivity target. However, the methods discussed are easily extensible to multiple sets of initial and target sets. None of the identified theoretical properties change for the more general case.

As mentioned earlier, we will perform such analysis using the equivalent representative planning model. To perform such analyses, we need to formalize the notion of model update. We start this formalization by first defining a model parameterization function that converts a given model to a set of features. In particular, we follow the conventions set in earlier model reconciliation papers \cite{model-rec-aij,model-rec-complexity} and define a function $\delta$ as follows.

\begin{defn}
The model parameterization function  $\delta$ maps a given model $\mathcal{M} = \langle F, A, I, G \rangle$ to a subset of propositions $\mathbb{P}$ (henceforth referred to as model parameters), where
\begin{align*}
  \mathbb{P} =~& \{\textit{init-has-f}\mid f \in F\} \cup\{\textit{goal-has-f}\mid f\in F\}\ \cup\\
  & \bigcup_{a \in A}\{a\textit{-has-prec-}f, a\textit{-has-add-}f, a\textit{-has-del-}f  \mid f \in F\}.
\end{align*}
The parameterization function $\delta(\mathcal{M})$  is defined as
\begin{align*}
    \tau_{I} &= \{ \textit{init-has-f}\mid f \in I\}\\
    \tau_{G} &= \{ \textit{goal-has-g}\mid g \in G\}\\
    \tau_{\textit{pre}(a)} &= \{ a\textit{-has-prec-}f \mid f \in \textit{pre}(a)\}\\
    \tau_{\textit{add}(a)} &= \{ a\textit{-has-add-}f \mid f \in \textit{add}(a)\}\\
    \tau_{\textit{del}(a)} &= \{ a\textit{-has-del-}f \mid f \in \textit{del}(a)\}\\
    \tau_a &= \tau_{\textit{pre}(a)} \cup  \tau_{\textit{add}(a)} \cup  \tau_{\textit{del}(a)}\\
    \tau_{A} &= \bigcup_{a\in A^\mathcal{M}} \tau_a\\
	\delta(\mathcal{M}) &= \tau_{I} \cup \tau_{G} \cup \tau_{A}
\end{align*}
\end{defn}

We use notation $\delta^{-1}$ to map subsets of $\mathbb{P}$ to possible model. To keep the potential space of model updates finite, we only consider changes that remove components and their parameters from the model representation. We examine model updates that disallow previously applicable plans, in other words, changes that constrain the model. We define a constrained version of a model as follows.
\begin{defn}
    A model $\mathcal{M}' = \langle F',I',A',G'\rangle$ is  a constrained version of a model $\mathcal{M}$, if $\delta(\mathcal{M}') \subseteq \delta(\mathcal{M})$ and there exists no action sequence $\pi$ such that $\Gamma_{\mathcal{M}}(I, \pi) \not\subseteq G$ and $\Gamma_{\mathcal{M}'}(I', \pi) \not\subseteq G'$.
\end{defn}

However, finding a potentially constrained version of a model by looking at potential model changes can be an expensive process.
Previous works try to simplify this search by assuming access to a limited set of hardening actions \cite{model-rec-aij, caglar_2024_can, keren_2019_goal}. 
In most network design settings, the cost of potentially updating a part of the model, hence changing the network, may be under-specified and, in some cases, unspecified.
This means we will need to consider the possibility of changing any and all parts of the network (and, by extension, the corresponding model description).
Unfortunately, for most realistic networks, the possible space of such changes is too large to be explored.
However, we show that by focusing on the problem of identifying constrained versions of the model, we can ignore all possible model updates except the ones that are part of the set, called {\em constraining changes set} or $\kappa(\mathcal{M}) = \tau_{I} \cup (\bigcup_{a\in A} \tau_{\textit{add}(a)})$. Now we demonstrate the effectiveness of $\kappa$ by showing that one can only create a constrained version of a model by using changes that are part of the constraining changes set.

\begin{prop}
If a model $\mathcal{M}' = \langle F',I',A',G'\rangle$ is a constrained version of $\mathcal{M} = \langle F,I,A,G\rangle$ such that, $F' = F$ and $\delta(\mathcal{M}') \subseteq \delta(\mathcal{M})$, then $\delta(\mathcal{M}) \setminus \delta(\mathcal{M}') \subseteq \kappa(\mathcal{M})$.
\end{prop}

This follows directly from the properties of models with positive pre-condition only. There is no case where a previously executable plan will be made inexecutable after removing a pre-condition, delete effect, or goal.

Note that our objective here is not just to find any constrained model but rather one that disallows potential attack paths while still preserving paths to nodes of interest. So, we find such models for both metrics by performing a search over the set of possible updates. We do this by employing a modified form of $A^*$-search 
, a systematic heuristic search. $A^*$ is guaranteed to return an optimal solution (i.e., the least costly set of model updates) as long as the heuristic is an optimistic estimate of the cost of the remaining part of the solution. Specifically, the $A^*$ search maintains a priority queue that is ordered by the sum of the cost of the current search node and the heuristic value associated with getting to a goal node.
In our case, the search node corresponds to a potential model you can obtain by applying the possible changes. The cost relates to the cost of changes already applied, and the heuristic approximates the cost of changes that need to be further applied to reach a model where some pre-specified goal condition related to the metrics is met. At any point in the search, the algorithm expands the node with the smallest possible value. If the selected node is a goal node, the search ends; if not, it adds all successor nodes (in our case, generated by applying available model updates), which are then added to the priority queue.  
In terms of updates to the search algorithm, we introduce changes that are applicable to both metrics, while still maintaining the guarantees provided by $A^*$. We also propose a new heuristic function that will help speed up the search in any context and focus on model updates that constrain the original model. 

For the heuristic, we start by assuming that all model updates have a uniform cost of one. This means that for any set of search nodes that is not the goal node, the heuristic is automatically admissible if the value is less than or equal to one (after all, it would take at least one more model update for it to meet the goal condition). However, setting all nodes to have equal heuristic values does not allow us to find more effective solutions. We would want to set lower heuristic values for more promising nodes. Specifically, for a given search node, we will generate a set of attack paths and see which one eliminates the most attack paths among all the successor nodes. 
Specifically, we will give a heuristic value proportional to the number of still valid attack paths for each successor. 
While one would ideally want to consider all attack paths, this may not be computationally feasible as we need to perform this on every node expanded by the search algorithm. Thus, we restrict it to some $K$ attack paths. Algorithm \ref{heur} presents a pseudo-code for calculating these heuristic values, and we show the effectiveness of this heuristic through empirical evaluation. 

\begin{algorithm}[tbp!]
\scriptsize
\caption{Heuristic function}
\label{heur}
\vspace{2pt}
\SetKwInOut{Input}{Input}
\SetKwInOut{Output}{Output}
\Input { $\hat{\mathcal{M}}$, plan\_list}
\Output { heuristic\_value}
\vspace{2pt} 
\vspace{2pt}
h $\leftarrow$ $|$plan\_list$|$;\\
\For{$a\_plan \in plan\_list $} {
\If{a\_plan is not valid in $\hat{\mathcal{M}}$ }{
$h -= 1$\;
}
}
heuristic\_value $\leftarrow$ h / $|plan\_list|$\; 
\textbf{return}  heuristic\_value
\end{algorithm}

The next optimization we introduce into the algorithm is node pruning. In particular, we want to avoid adding nodes to our priority queue that are guaranteed never to lead to a goal. Note that at this point, we have yet to define our goals precisely.
Regardless of the exact form the function related to attack paths takes, in each case, 
we would want $\mathcal{F}_{\mathcal{C}}^{1}$ to return a value of one. This means we can prune out any search nodes that cannot lead to models where there are connectivity paths.

\begin{prop}
Let $\mathcal{M}$ be an equivalent representation of \frm\ $\mathcal{G}$, an initial node set $S_N^0$ and a set of connectivity target $S_t^{\mathcal{C}}$, such that there exists no valid connectivity path for the graph, then there exists no constrained version of $\mathcal{M}$, where the corresponding \frm\ model has a valid connectivity path.
\end{prop}

This proposition directly follows from the fact that a constraining model update never adds new plans, and a valid plan must exist for any valid connectivity path. This result tells us that we can ignore any successors of a search node where no valid path to a connectivity target exists. 

Now that we have covered two modifications that are designed to exploit the structure of the problem to improve the effectiveness of the search process, we 
extend $A^*$ to handle one of the shortcomings of the setting, namely missing modification costs. $A^*$ is generally deployed in settings where each step that can be taken from a search node has a well-defined cost. All the theoretical guarantees of $A^*$ are contingent on access to a specified cost function. In our setting, each potential step corresponds to some update in the model, which further corresponds to a change in the underlying network. This might include changes ranging from potentially patching software to potentially removing network connections. Humans are generally known to be bad at approximating costs and utilities in general \cite{booth_2023_perils}, however, such problems are made worse in settings like network administration, where there are many competing interests.
For example, even a small change in updating a software version may involve retraining a number of users.  

Given these considerations, it is unlikely that the system typically has access to a fully defined cost function associated with the possible network modifications. As such, the search would not be able to identify a single optimal modification set.
Instead, we focus on identifying options that meet the required goal conditions.
These options can then be presented to the defender, who can then compare the modifications being proposed by each option and select one they might believe is the most reasonable one. 
When coming up with these options, our primary focus would be to choose solutions that make use of diverse network modifications.
The use of diverse solution sets as a means of addressing incompletely specified cost functions is a strategy that has been shown to be effective in many contexts \cite{ghasemi_2021_multiple}. 
In this setting, by choosing diverse modifications, we increase the chance that one of the solutions consists of modifications that the user may deem to be less costly.

In particular, we will come up with a set of updates such that no set of model updates is a subset of another. We will call such a set of model updates a non-overlapping set.
\begin{defn}
\label{defn:set}
    A set $\mathbb{C} = \{\mathcal{C}_1,..., \mathcal{C}_k\}$ is a {\em non-overlapping set} of model updates, if there exists no $1\leq j\leq k$ and $j\neq i$, such that $\mathcal{C}_i \subseteq \mathcal{C}_j$.
\end{defn}

We will generate such a non-overlapping set by extending $A^*$ to identify multiple solutions. Effectively, each time a goal condition is met, it will get added to a set of solutions, and no successor node that is a super-set of an already found solution is added to the search priority queue. Here, we impose a search budget that ensures the search exits after a certain number of search nodes are expanded. 
Algorithm \ref{algo1} presents the modified $A^*$ search that incorporates all the changes discussed above. 

\begin{prop}
   Algorithm \ref{algo1} will always return a solution set that is non-overlapping.
\end{prop}
This proposition holds true, as the search is always guaranteed to test potential solutions in the order of increasing cardinality. This means we only need to ensure that any future solution is not a superset of a previous solution to ensure that the solution set satisfies Definition \ref{defn:set}.

\begin{algorithm}[tbp!]
\scriptsize
\caption{Search for non-overlapping set}
\label{algo1}
\vspace{2pt} 
\SetKwInOut{Input}{Input}
\SetKwInOut{Output}{Output}
\Input { $\mathcal{M} = \langle F, A, I, \mathbb{G}_C, \mathbb{G}_R \rangle$, budget $k$}
\Output { Solution Set $\hat{\mathbb{K}}$}
\vspace{2pt} 
fringe $\leftarrow$ \texttt{PriorityQueue()}\;
c\_list $\leftarrow \langle\rangle$ {\em (Closed list)}\;
cnt $\leftarrow 0$\;
solution\_list$ \leftarrow \langle\rangle$ \;
$\text{fringe.push}(\{\}, 0)$\;
\vspace{2pt} 
\While{\text{fringe not empty} \& $cnt \leq k$}{
\vspace{2pt}
$cnt += 1$\;
$\hat{\mathcal{C}} \leftarrow \text{fringe.pop}()$\;
$\hat{\mathcal{M}} \leftarrow \delta(\mathcal{M})\setminus \hat{\mathcal{C}}$
\If{goal condition met for $\hat{\mathcal{M}}$ } {
solution\_list.add($\hat{\mathcal{C}}$)\;
}
 c\_list.add($\hat{\mathcal{C}}$)\;
\vspace{2pt} 
\For{$f \in \kappa(\mathcal{M})$} {
\If{$\{f\} \cup \hat{\mathcal{C}}$ has no subset in $\hat{\mathbb{C}}$ and there exists a connectivity path}{
$\hat{\mathcal{C}}' \leftarrow  \{f\} \cup \hat{\mathcal{C}}$;

$\hat{\mathcal{M}}' \leftarrow \delta(\mathcal{M})\setminus \hat{\mathcal{C}}'$;

$\Pi_{\mathcal{C}} \leftarrow \texttt{FindTopAttackPlans}(\mathcal{M})$

$\text{fringe.push}(\hat{\mathcal{C}}', |\mathcal{C}'| + \textrm{heuristic}(\hat{\mathcal{M}}',\Pi_{\mathcal{C}} ))$\;
}
}
}

\end{algorithm}

\subsubsection{Goal Condition}
Now, one can adapt Algorithm \ref{algo1} for each individual metric by changing the goal condition. For the impenetrability metric, the most intuitive condition would be to have both $\mathcal{F}_{\mathcal{A}}^{\mathcal{I}}$ and  $\mathcal{F}_{\mathcal{C}}^{\mathcal{1}}$ return one.
To identify such models, we need to check that the model does not return any plan for the attack targets and that there still exists a plan for the connectivity target. This can be extended to multiple initial states and goals either by checking pairwise or by using compilation techniques similar to the one discussed below.

However, for attack difficulty, we have a slightly more complex goal condition. We need to make sure that $\mathcal{F}_{\mathcal{A}}^{\mathcal{D}}$ meets certain threshold. 
The particularly interesting case here is one where we have multiple initial state sets and multiple targets.
This means that we need a way of finding the shortest possible attack path from a set of possible starting states to the set of possible attack targets. 
We will do so by employing a planning compilation that will automatically identify this path.

Our objective here is to build on our planning model to support finding the shortest possible path from multiple initial state sets and multiple targets. 
We will represent this updated model as $\CompiledModel = \langle F^{\CompiledSymb}, A^{\CompiledSymb}, I^{\CompiledSymb}, G^{\CompiledSymb}\rangle$.
The new fluents $F^{\CompiledSymb}$ is given as
\( F^{\CompiledSymb} =  F \cup 
\{\texttt{init\_change\_mode}\} \cup 
\{\texttt{act\_mode}\} \cup \{\texttt{goal\_reached}\}  \)
Here $F^{\CompiledSymb}$ includes all the original fluents, plus three new fluents. 
First, we have the proposition {\texttt{init\_change\_mode}}, which enables the planner to select the starting initial state set. Next, the fluent {\texttt{act\_mode}} indicates that the planner can only select the attack actions, and finally, the fluent {\texttt{goal\_reached}} indicates that the intended target has been achieved.
Our new set of actions  $A^{\CompiledSymb}$ is as follows:
\(A^{\CompiledSymb} =  A^O \cup A^I \cup A^G \). \\
This set comprises the following components. Firstly, it includes all the original actions of the attacker, denoted as $A^O$, but with an added pre-condition {\texttt{act\_mode}}. This pre-condition is necessary to ensure that the attack actions are only executed in the appropriate mode of operation. Furthermore, we introduce a new set of actions, $A^I$, specifically designed to transition the state to one of the potential initial states. Another subset of actions, $A^G$, is introduced to trigger the proposition {\texttt{goal\_reached}} when the goal conditions of the attacker are met.
More specifically, for all possible initial states $S_i$ we have corresponding action $a^{I}_i = \langle \textit{pre}(a^{I}_i),  \textit{add}(a^{I}_i), \textit{del}(a^{I}_i)\rangle$,  such that
\begin{align*}
\textit{pre}(a^{I}_i) = \{\texttt{init\_change\_mode}\},\\
\textit{add}(a^{I}_i) = \gamma^{\mathcal{M}}(S_i) \cup \{\texttt{act\_mode}\},\\ \textit{del}(a^I_i) = \{\texttt{init\_change\_mode}\}
\end{align*}
The delete effect in the initial state action ensures that it can only select an initial state once and prevents the planner from possibly creating infeasible paths by adding elements from other initial states. Once the initial state action is executed, the fluent \texttt{act\_mode} is made true, and now the planner can execute the attack actions.

For all possible reachable attack targets $(h_i, p_i)$, we introduce new actions $a^G_i = \langle \textit{pre}(a^G_i),  \textit{add}(a^G_i), \textit{del}(a^G_i)\rangle$ with the corresponding definition
\begin{align*}
\textit{pre}(a^G) = {\gamma^{\mathcal{M}}(h_i, p_i)},\\
\textit{add}(a^G) = \{\texttt{goal\_reached}\}, \textit{del}(a^G) = \{\}
\end{align*}

The initial state consists of only the {\texttt{init\_change\_mode}} fluent
\[I^\CompiledSymb ={\{\texttt{init\_change\_mode}\}}\]
Lastly, the goal becomes
\[G^\CompiledSymb =  \{\texttt{goal\_reached}\}\]
An optimal plan for such a model will be one in which the planner chooses an initial state and goal pair that results in the attack plan with minimal cost.

\begin{prop}
    For a given graph $\mathcal{G}$, a set of initial state sets $\mathbb{S}^0_\mathcal{N}$ and attack targets, $\mathbb{S}^\mathcal{A}_t$, such that a valid attack path exists, the value returned by $\mathcal{F}^\mathcal{D}_\mathcal{A}$, is equal to the cost of the optimal plan for the corresponding compiled planning model $\CompiledModel = \langle F^{\CompiledSymb}, A^{\CompiledSymb}, I^{\CompiledSymb}, G^{\CompiledSymb}\rangle$.
\end{prop}
In the compiled model, the planner has the freedom to choose the initial state and the goal state, but it must always select one. 
Since we are using the optimal planner, it will always choose the initial and goal state that will result in the cheapest possible attack path. Because every valid plan in this domain must include a single initial state selection action and a target selection action. As such, the only place the planner can optimize is in the cost of the attack path. Thus it returns $\min \{|\mathbf{E}|\mid \mathbf{E} \in \mathbb{E}_{\mathcal{A}}\}$.

In this case, the goal check simply involves calling this compilation. If the cost of the minimal plan is greater than or equal to the threshold, then the identified solution is added to the solution list. 

%% file: Sections/5-architecture.tex
\section{The \tool Toolset}
\label{sec:architecture}

The overview of the \tool tool is illustrated in Figure \ref{architecture}. \tool is integrated with an off-the-shelf network vulnerability scanner that generates a list of vulnerabilities and their descriptions. 
\begin{figure}[!t]
\centering
\includegraphics[width=\columnwidth]{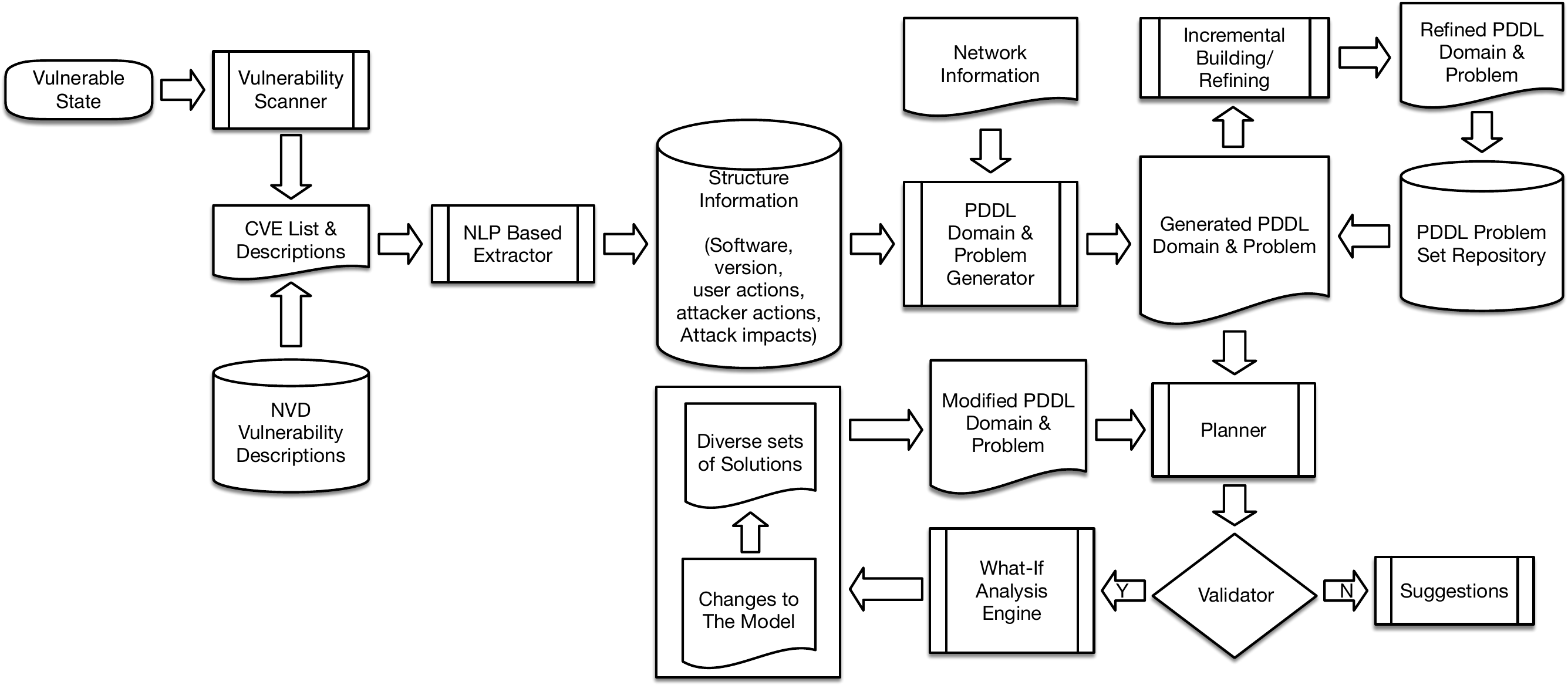}
\caption{\tool Tool Architecture.}
\Description{Overview or \tool.}
\label{architecture}
\end{figure}

\input{Sections/3-example}

 We collected detailed system and network data, categorized as host-level and connectivity attributes. These were transformed into a formal representation using a domain and problem generator, enabling automated planning.
Once generated, both the domain and problem files are introduced to a \textit{Validator}. This \textit{Validator} scrutinizes the plan for a designated goal. A valid plan signifies the existence of a pathway that an attacker could exploit to compromise the target or achieve a specific goal state. If such pathways or plans are discerned, the domain and problem files are then relayed to a \textit{what-if-analysis} engine.
For \textit{impenetrability metric}, the analysis engine conducts an exhaustive exploration, by constrained model updates, to find all potential sets of changes that would obstruct any plans leading to the goal state. Each derived set of changes culminates in a solution, which translates into a revised domain and problem file. For each unique set of changes, or, in other words, for each distinct solution, the \textit{Planner} indicates that the original plans are invalid. Delving deeper, the \textit{Validator} identifies and delineates the specific initial states where the actions prove ineffective. Lastly, the \textit{Validator} genrate suggestions which include both the failed plan and its associated unsuccessful actions, to providing users with insights into the modifications required within network. 
For, \textit{attack difficulty metric}, we are finding the set of changes for which the minimum cost for the attacker to reach the target node increases in the \textit{Planner}. Considering that the weakness of a network depends on its most vulnerable attack path, this analysis significantly increases the difficulty for attackers within a specific network.

Here we are introducing \textit{human-in-the-loop} concepts where the administrator would have all possible sets of diverse changes and their corresponding solutions. They can choose any possible change. Also, if there are any new modifications in the network, we introduced continuous building and refining options for updating the previous model copies.

%% file: Sections/3-example.tex
\begin{figure}[!t]
  \centering
  \includegraphics[width=\columnwidth]{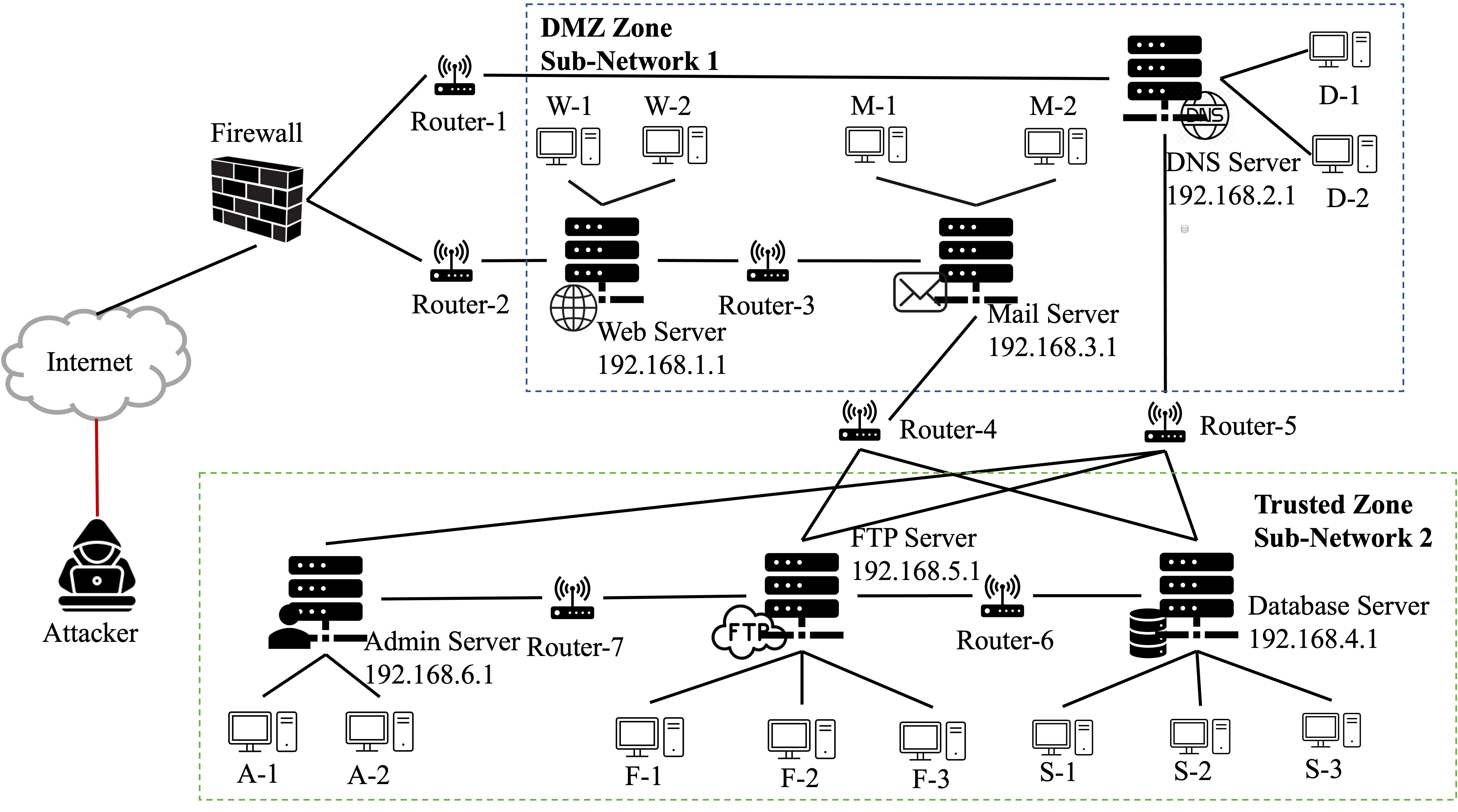}
  \caption{Topology of Test Network.}
  \Description{Topology of Test Network.}
  \label{fig:plan}
\end{figure}

Figure~\ref{fig:plan} shows the test network we used in this study. We created the test network using various virtual environments and tools to simulate real-world scenarios. This setup included virtual machines, network configurations, and firewall rules. The network is divided into two subnetworks: Sub-Network 1 (demilitarized zone) that contains DNS, Mail \& Web Server, and Sub-Network 2 (trusted zone) that contains FTP, Database \& Admin Server running various services.
TABLE \ref{tab:vul} (in Appendix B) lists the vulnerabilities present in this network. Before \tool is used, we gather information on network topology \& connectivity, vulnerabilities, and services. \tool will create the \asg\ based on this information. A visual representation of the \asg\ is shown in Figure \ref{fig:acg-graph} in Appendix A.

%% file: Sections/6-case-study.tex
\section{What-if Analysis Using \tool}
\label{sec:case-study}
In this section, we will discuss how our tool can be used in the context of the network presented in Figure~\ref{fig:plan} to perform \textit{what-if analysis}. The reader may want to refer to the corresponding \frm\ shown in Figure \ref{fig:acg-graph} in Appendix A.
\subsection{Impenetrability Analysis}
\label{unsolvable}
 Our first analysis is based on \textit{impenetrability metric} (see Definition 8). Typically, the system administrator starts by selecting the network node they would like to focus on, say the Admin Server. Since \tool\ models the attack-connectivity graph as a planning model, we can readily use existing solvers to present the users with a set of diverse attack vectors of different qualitative characteristics (via forbid-iterative). In this case, we found 21 different sets of attack paths to compromise the Admin Server. 

However, the goal of our tool is to go beyond such shallow analysis and help the user directly identify ways in which attack paths to the node can be disrupted while ensuring that the node is still accessible from other nodes in the network.
Using our model space search technique, \tool\ identifies unique ways in which the security stances of the node in question may be hardened. Each of these different updates may pose unique challenges to implement, which may be hard to implement but easier for a domain expert to assess. In this case, the search identifies a set of diverse possible changes shown in Figure~\ref{changes}.

\begin{figure}[h!]
    \fontsize{9}{10}\selectfont
    \centering
        \begin{minipage}{\linewidth}
            \begin{lstlisting}
; SPEAR Tool Output:
Changes -> {'has-initial-state-admin-driver-1 gmetad-slash-server-dot-c', 'has-initial-state-admin-driver-2 mpm-event-worker-or-prefork'}
Changes -> {'has-initial-state-adminserver-system-1-config ganglia-3-dot-1-dot-1', 'has-initial-state-admin-driver-2 mpm-event-worker-or-prefork'}
Changes -> {'has-initial-state-vul-admin-function-1 process_path-function', 'has-initial-state-admin-driver-2 mpm-event-worker-or-prefork'}
Changes -> {'has-initial-state-adminserver-system-2-config apache-http-server-2-dot-4-dot-17-to-2-dot-4-dot-38', 'has-initial-state-admin-driver-1 gmetad-slash-server-dot-c'}
            \end{lstlisting}
        \end{minipage}
    \caption{Changes suggested by \tool Tool.}
    \Description{Changes suggested by \tool Tool Topology of Test Network}
    \label{changes}
\end{figure}

Note that, given the level of abstraction used in the modeling, the system is able to give suggestions at the level of potential changes in network topology or specific vulnerabilities. The question of how to implement these changes is left to the administrator. Once presented with the options, the administrator is free to consider each option and validate it by implementing those network changes.

If the administrator selects $2^{\text{nd}}$ change (transalted in natural language by \tool) -- \textit{Upgrade the ganglia software's version anything other than v3.1.1 to a more recent version, and modify the MPM event worker or prefork driver} -- shown in Fig.~\ref{changes}, then the above attack plans would fail.
These case study proves that the \textit{impenetrability metric}, $\mathcal{F}_{\mathcal{A}}^{\mathcal{I}} (\mathbb{E}_{\mathcal{A}}) = 1$ as there are no valid attack paths ($|\mathbb{E}_{\mathcal{A}}|=0$) to compromise Admin Server. The administrator is now free to investigate other potential attack paths or even sketch a potential attack path and see how it might fail in the context of a given model update. The number of changes is not limited to 4, as the planner can generate more diverse sets of suggested changes for other nodes that could potentially disrupt the attacker from reaching the Admin Server.
By opting for any of these model changes, all the attack paths targeting Admin Server are now invalid but the connectivity paths or services to that Admin Server are still valid ($\mathbb{E}_{\mathcal{C}}$ is valid), so the connectivity function, $\mathcal{F}_{\mathcal{C}}^{1} (\mathbb{E}_{\mathcal{C}}) = 1$. Figure \ref{fig:connectivity} shows a valid plan to reach the Admin Server when the model changes are applied and there exists no plan to compromise the Admin Server.

\begin{figure}[h!]
    \fontsize{9}{10}\selectfont
    \centering
        \begin{minipage}{\linewidth}
            \begin{lstlisting}
(has_access_to_mail_server web-server mail-server internet gateway-192-168-3-1 p-25 smtp)
(has_access_to_ftp_server_via_mail_server   mail-server ftp-server internet gateway-192-168-5-1 p-20 tcp)
(has_access_to_admin_server_via_ftp_server ftp-server admin-server internet gateway-192-168-4-1 p-1311 https)
            \end{lstlisting}
        \end{minipage}
    \caption{Plan (Connectivity Path $ \mathbb{E}_{\mathcal{C}}$) to reach Admin Server.}
    \Description{Plan (Connectivity Path $ \mathbb{E}_{\mathcal{C}}$) to reach Admin Server when there are no Attack Path.}
    \label{fig:connectivity}
\end{figure}

When conducting our \textit{what-if analysis}, we factored in the ``human-in-the-loop'' perspective. For the sake of simplification, we initially assigned a unit cost for each change. However, in real-world applications, this unit cost might vary based on the nature and impact of the change.
For instance, while the tool might view all sets of changes as having the same unit costs, the actual implications for an organization or administrator might differ significantly. For example, for changes to \textit{MPM event worker or prefork} and \textit{Ganglia version 3.1.1}, any changes to the former might be perceived as more cumbersome and thus more costly than updating the latter to a newer version.
Consequently, the administrator might prefer to implement the $2^\text{nd}$ set of changes over the $3^\text{rd}$ set, even if its unit cost is the same. This decision underscores the importance of accounting for the nuanced and often unpredictable nature of human decision-making in our analysis.

\subsection{Attack Difficulty Analysis}
\label{minpath}

The other notable application of the \tool tool is its capability to augment the minimum cost associated with all possible initial and target states from an attacker's perspective. In this analysis, we are focusing on metric 2 or \textit{attack difficulty metric} (Definition 9). Considering that a network's robustness is contingent upon its weakest attack path, this analysis effectively elevates the challenge for attackers in a given network. Suppose an administrator aims to escalate the minimum cost of the attack path under various potential scenarios, which constitute possible initial states for attackers:
\begin{enumerate}
    \item Attacker has access to the DNS server's vulnerabilities.
    \item Attacker has access to the Mail server's vulnerabilities.
    \item Attacker has access to SQL server's vulnerabilities.
\end{enumerate}

Based on these initial states we found that the attacker can compromise the Admin Server, provided the attacker's initial state is established subsequent to the exploitation of CVE-2017-14491 (i.e., attack path $ \mathbb{E}_{\mathcal{A}}$)  with an \textit{attack difficulty metric},  $\mathcal{F}_{\mathcal{A}}^{\mathcal{D}} (\mathbb{E}_{\mathcal{A}}) = 3$. Now, let us assume the administrator seeks to identify model updates that would result in a minimum increase of 2 unit costs. 
\begin{figure}[!ht]
    \fontsize{9}{10}\selectfont
    \centering
        \begin{minipage}{\linewidth}
            \begin{lstlisting}
; SPEAR Tool Output:
Changes -> {'has-initial-state-has-adminserver-system-2-config apache-http-server-2-dot-4-dot-17-to-2-dot-4-dot-38'}
Changes -> {'has-initial-state-has-admin-driver-2 mpm-event-worker-or-prefork'}
Changes -> {'has-initial-state-has-dns-driver-2 request-handling', 'has-initial-state-has-dnsserver-system-1-config dnsmasq-before-2-dot-78'}
Changes -> {'has-initial-state-has-dnsserver-system-1-config dnsmasq-before-2-dot-78', 'has-initial-state-has-dnsserver-system-2-config windows-dns-server'}
            \end{lstlisting}
        \end{minipage}
    \caption{Changes suggested by \tool to increase the minimum cost of the given initial states and all possible targets.}
    \Description{Changes suggested by \tool to increase the minimum cost of the given initial states and all possible targets.}
    \label{changes_2}
\end{figure}

By employing our model space search technique, the administrator can select the modification depicted in Figure \ref{changes_2} to guarantee an increase (at least 2 unit cost) in the cost of compromising one of the possible targets. Specifically, if an administrator implements the first set of suggestion shown in Figure \ref{changes_2}, the minimum cost of the attack path ($\mathbb{E}_{\mathcal{A}}$) increases, and the new \textit{attack difficulty metric},  $\mathcal{F}_{\mathcal{A}}^{\mathcal{D}} (\mathbb{E}_{\mathcal{A}}) = 5$. Note, in this scenario, the minimum cost of the attack path is now associated with compromising the Mail Server following the implementation of the model updates. 
This method ensures that all previously optimal attack plans become invalid due to the increased cost. Consequently, any new plans developed post-update will incur a cost higher than that of the earlier plans.

%% file: Sections/7-eval.tex
\section{Empirical Evaluation}
\label{sec:evaluation}
Our primary focus is to evaluate the performance characteristics of \tool changed with respect to the complexity of network and attacker goals. 
We ran our experiment on a 3.1 GHz Dual-Core Intel Core i5 processor, with 8 GB 2133 MHz, and LPDDR3 memory.

\input{Files/tables/run_ex_apr_15}

We first investigate how the performance of the \tool tool changes with respect to the computation time as we vary the target (goal) node. The purpose of this experiment is to understand the tool's performance for  Impenetrability and Attack Difficulty Analysis.
TABLE \ref{tab:running_example_time} presents the time required to identify two distinct solutions for each potential goal or attacker's target. For this analysis, the value of $\alpha$  is set to 2 (defining the threshold for considering an improvement in robustness; a solution is deemed to enhance this security only if it increases by at least 2), enabling the examination of solutions under metric 2 (\textit{attack difficulty metric}). 

In our example, solutions are simple, usually requiring only one or two steps. Computing multiple plans takes longer, and in some cases, skipping the heuristic and selecting updates randomly led to faster results. Thus, the time saved by the heuristic doesn’t always outweigh the cost of extra plan calculations.
To explore the effectiveness of our approach in more complex scenarios, we evaluated it by varying the number of nodes and vulnerabilities.

\input{Files/tables/upd_table}

The second factor we investigate is the computational
characteristics of our proposed algorithm when we increased the number of nodes in the network to address the scalability issues in network security. We generated directed, scale-free network topologies using the Barabási-Albert model to simulate realistic network environments for \textit{what-if analysis}, as this model reflects the power-law structure commonly found in real-world networks.
We varied the total number of nodes (10,15,20,25,30) and vulnerabilities associated with each node, and each new node added during the network's construction connects to 2 existing nodes. 

TABLE \ref{tab:table_time} displays the time taken (in seconds, along with their standard deviations) to find two distinct solutions for metrics 1 and 2, we evaluated results with and without using heuristics, averaging across varying numbers of vulnerabilities per node group. Specifically, each node group is analyzed across four problem instances with vulnerabilities representing 20\%, 40\%, 60\%, and 80\% of the total number of nodes, with vulnerabilities randomly distributed. For example, a graph with 10 nodes includes problem instances with 2, 4, 6, and 8 randomly distributed vulnerabilities. While metric 1 uses a fixed initial and goal node for each problem instance, metric 2 employs three randomly selected initial nodes and three randomly chosen goal nodes, and the $\alpha$ value for metric 2 is set to 2. As observed from the table, using heuristics reduces the time required to find solutions for both metrics across various node counts.

The search utilizes modified $A^*$ to identify model updates, it is crucial to prioritize which updates to apply first. The heuristic we use effectively manages this prioritization. Specifically, for both metrics 1 and 2, each current model update begins with the calculation of three different plans, based on which the heuristic evaluations are made (as detailed in Algorithm 2). This process provides insights into how close the subsequent updates are to achieving the desired final outcomes.

%% file: Files/tables/run_ex_apr_15.tex
\begin{table}[!ht]
\caption{Comparison of  Impenetrability (M1) and Attack Difficulty (M2) Analysis, with and without Heuristics.}
\centering
\begin{footnotesize}
\begin{tabular}{|c|rrll|}
\hline
\multicolumn{1}{|l|}{\multirow{2}{*}{\textbf{ \# Nodes}}} & \multicolumn{4}{c|}{\textbf{Time (s)}}   \\ 
        \cline{2-5} 
\multicolumn{1}{|l|}{} &
  \multicolumn{1}{c|}{M1} &
  \multicolumn{1}{c|}{M1 + Heuristic} &
  \multicolumn{1}{c|}{M2} &
  \multicolumn{1}{c|}{M2 + Heuristic} \\ \hline
Admin                                        & \multicolumn{1}{r|}{20.96}  & \multicolumn{1}{r|}{23.54}  & \multicolumn{1}{l|}{1540.27}  & 1493.88  \\ \hline
DNS                                          & \multicolumn{1}{r|}{77.75}  & \multicolumn{1}{r|}{29.14}  & \multicolumn{1}{l|}{447.39}   & 843.68   \\ \hline
FTP                                          & \multicolumn{1}{r|}{169.75} & \multicolumn{1}{r|}{246.86} & \multicolumn{1}{l|}{20335.41} & 31762.01 \\ \hline
Mail                                         & \multicolumn{1}{r|}{9.94}   & \multicolumn{1}{r|}{14.79}  & \multicolumn{1}{l|}{145.18}   & 492.45   \\ \hline
Database                                          & \multicolumn{1}{r|}{131.93} & \multicolumn{1}{r|}{260.23} & \multicolumn{1}{l|}{300.13}   & 981.25   \\ \hline
Web                                          & \multicolumn{1}{r|}{7.43}   & \multicolumn{1}{r|}{6.96}   & \multicolumn{1}{l|}{103.41}   & 509.33   \\ \hline
\end{tabular}
\end{footnotesize}
\label{tab:running_example_time}
\end{table}

%% file: Files/tables/upd_table.tex
\begin{table}[h]
\caption{
Comparison of Time Required for Various Nodes.}
\centering
\resizebox{\columnwidth}{!}{%
\begin{tabular}{|l|llll|}
\hline
\multirow{2}{*}{\textbf{\# Nodes}} & \multicolumn{4}{c|}{\textbf{Time (s)}}                        
\\ \cline{2-5} 
& \multicolumn{1}{c|}{M1} & \multicolumn{1}{c|}{M1+Heuristic} & \multicolumn{1}{c|}{M2} & \multicolumn{1}{c|}{M2+Heuristic} \\ \hline
10 & \multicolumn{1}{l|}{260.78 $\pm$ 149.72}   & \multicolumn{1}{l|}{213.55 $\pm$ 129.21}   & \multicolumn{1}{l|}{97.2 $\pm$ 92.07}       & 88.67 $\pm$ 97.22    \\ \hline
15 & \multicolumn{1}{l|}{348.57 $\pm$ 299.74}   & \multicolumn{1}{l|}{260.01 $\pm$ 251.46}   & \multicolumn{1}{l|}{526.36 $\pm$ 479.95}    & 289.15 $\pm$ 345.79  \\ \hline
20 & \multicolumn{1}{l|}{4389.53 $\pm$ 3210.5}  & \multicolumn{1}{l|}{2664.91 $\pm$ 1616.98} & \multicolumn{1}{l|}{673.04 $\pm$ 690.33}    & 216.13 $\pm$ 147.66  \\ \hline
25 & \multicolumn{1}{l|}{1293.94 $\pm$ 1224.48} & \multicolumn{1}{l|}{244.68 $\pm$ 147.07}   & \multicolumn{1}{l|}{8071.04 $\pm$ 12964.31} & 3475.81 $\pm$ 4387.2 \\ \hline
30 & \multicolumn{1}{l|}{3266.08 $\pm$ 2712.53} & \multicolumn{1}{l|}{1680.93 $\pm$ 2123.6}  & \multicolumn{1}{l|}{7404.54 $\pm$ 11908.29} & 3756.9 $\pm$ 7200.47 \\ \hline
\end{tabular}%
}
\label{tab:table_time}
\end{table}

%% file: Sections/8-Related-Work.tex
\section{Related Work}
\label{sec:related-work}

Attack graphs (and variants) have been variously studied by the security research community for modeling and analyzing security postures of networked systems and drive cyber risk management (see for example,
~\cite{phillips-1998,ray-pool-2005,ou-mulval-2005,noel-2003,b17}). 
The main efforts have been in addressing issues such as automatically generating attack graphs and identifying attack paths
~\cite{kaynar2015distributed}, visualization of security postures \cite{b43},
 attack graphs as analysis tool
~\cite{zhao_2021_structural}, and for security hardening
~\cite{dewri_optimal_2007,dewri_optimal_2012,b18}.
 Poolsappasit et al.\ introduced Bayesian Attack Graphs (BAG)~\cite{poolsappasit2011dynamic} for risk evaluation and mitigation planning. Beckers et al.\ combined high-level attack tree analysis and low-level attack graph analysis to compute the probability of successful attacks~\cite{beckers2014determining}.
 This process was extended in Beckers et al. by considering the effects of social engineering threats~\cite{beckers2015analysis}. 

AI-Planning tools have been introduced in the cyber-security applications which use PDDL to generate attack graph \cite{ghosh2012planner,tiwary2017pddlassistant,b37,ray_2023_AI,bashir2024resiliency}. 
Ghosh et al. proposed  Planner~\cite{ghosh2012planner}, a specialized search algorithm from the field of AI, to generate scalable and time-efficient representations of attack graphs. Tiwary et al.\ introduced a methodology and a tool called PDDLAssistant~\cite{tiwary2017pddlassistant} to facilitate the incremental development of PDDL representations for cyber-security domains. This process is further enhanced in Bezawada et. al.\ where they introduced AGBuilder~\cite{b37} to address the limitations of scalable attack graph generation and representation.

Attack graphs have often proven to be versatile tools for enhancing the analysis and understanding of complex security vulnerabilities across various systems~\cite{unger2023risk}. Albanese et al.\ provided polynomial time algorithms~\cite{albanese2013efficient} that evaluated k-zero-day safety in large networks. However, this work is not suited for what-if analysis.  
Miller et al.\ introduced GRASP~\cite{miller_2023_grasp}, leveraging machine learning to optimize attack path identification, which significantly speeds up computational times without compromising on accuracy. Complementing this, Milani et al.\ ~\cite{milani_2020_harnessing}  delve into the strategic use of deception in security games.
However, a limitation was the need for additional memory page transfers among search agents and coordination among them. Our proposed tool \tool is able to handle the complexity and scalability issues as we are using an AI planner with heuristics. 

In the field of network hardening, researchers have made significant contributions, building upon each other's work to advance the understanding and application of this approach
~\cite{dewri_optimal_2007,dewri_optimal_2012}.
Dewri et al.\ \cite{dewri_optimal_2007} introduced the notion of optimal security hardening by developing attack tree cost models for network attack and defense. Later, they extended this work by considering the attacker's perspective in network hardening~\cite{dewri_optimal_2012}. 
Roberts et al.\ proposed an alternative approach to home computer security software using a software agent~\cite{roberts2012using} that could assist users in managing their security risks.
But, traditional attack graph analysis approach treats network hardening as a multi-objective optimization problem, focuses on producing an ``optimal'' solution based on predefined criteria. In contrast, \tool\ enables what-if analysis, allowing sys-admin  to explore multiple scenarios and choose the most desirable course of action based on cost, availability of resources, human intuition, or other constraints. Sometimes an optimal solution may not be a viable solution for a particular scenario. The key difference is that, \tool could incorporate human decision-making (via what-if analysis) rather than replacing it with automated optimization.

%% file: Sections/9-conclusion.tex
\section{Conclusion and Future Work}
\label{sec:conclusion-future}

In this work, we have presented a method for integrating AI planning techniques with attack graphs to enhance network security. Our contributions involve modeling vulnerabilities and their dependencies as well as network connectivity as a hypergraph called \frmtitle\ (\asg) using the PDDL. \asg\ is non-monotonic, unlike most other similar models in the prior art. We develop the \tool toolset for evaluating \frmtitle\, creating unsolvable plans for proposed changes. The tool allows a human-in-the-loop interaction. These features contribute synergistically to allow the administrator to perform what-if analysis and also to reason about different attacker motives. The logic and reasoning framework of \tool closely matches how human users reason about actions, and thus the suggestions generated are intuitive and easy to understand. 
Our future work includes improving the framework to allow for cost-benefit analysis for network hardening. We also plan to do usability studies to understand how suggestions provided by \tool are perceived by domain experts.

%% file: Sections/11-acknowledgment.tex
\section*{Acknowledgment}

This work was partially supported by the U.S.\ National Science Foundation under Grant No. 1822118 and 2226232,
Award Numbers DMS 2123761, by the industry member partners of the NSF IUCRC Center for Cyber Security Analytics and
Automation -– AMI, NewPush, Cyber Risk Research, NIST and ARL -– by the State of Colorado (grant \#SB 18-086) and by the
authors’ institutions. Any opinions, findings, and conclusions or recommendations expressed in this material are those of the authors and do not necessarily reflect the views of the NSF, or other organizations and agencies.

%% file: Sections/10-appendix.tex
\section*{Appendix A: Visual Representation of Attack-Connectivity Graph}

\begin{figure*}[hbtp]
    \centering
    \includegraphics[width=\textwidth]{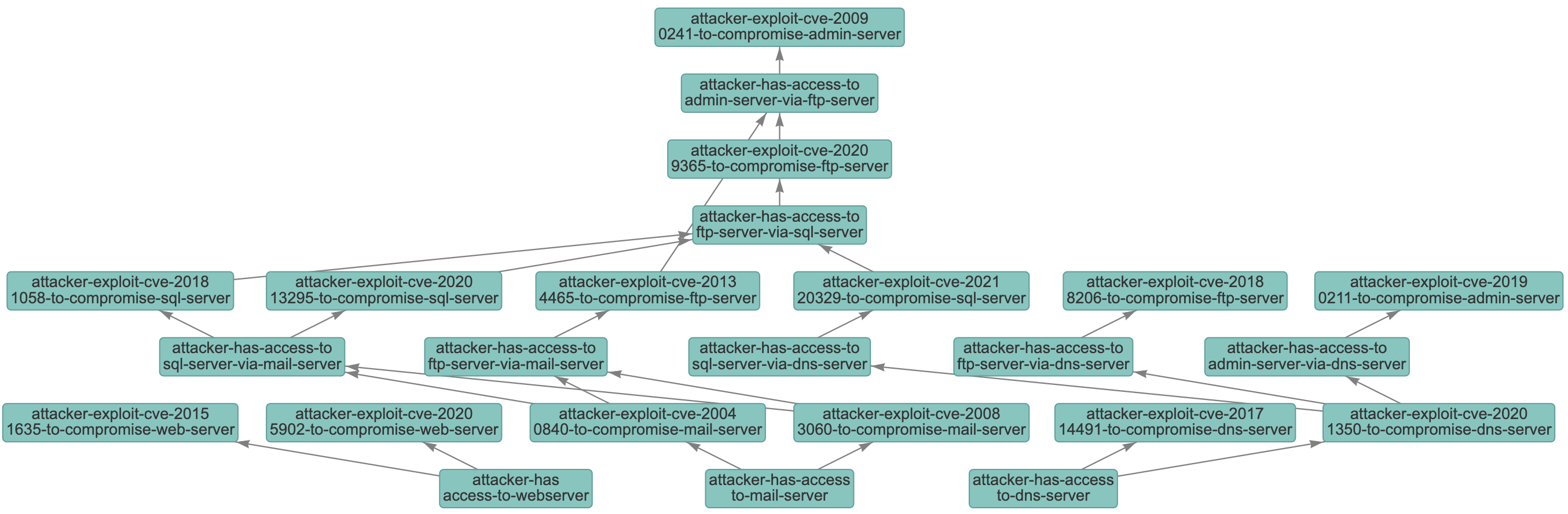}
    \caption{Attack-Connectivity Graph of Test Network.}
   \label{fig:acg-graph}
\end{figure*}

Figure \ref{fig:acg-graph} provides a visual representation of the \asg\ for the test network used in this work and described earlier in Section \ref{sec:architecture}. In \asg\, we have two types of directed edges. The normal directed edges indicate that the attacker exploits a particular vulnerability to compromise a node. The node/host remains functional after the attack, allowing the attacker to use it as a reference to compromise additional nodes. But the directed edges with a cross represent that an attacker exploits a particular vulnerability to compromise a node, rendering it nonfunctional post-attack. For example, an attacker might exploit a heap-based buffer overflow vulnerability in a DNS Server machine (D-1) to launch a DoS attack, which then causes the machine to crash and become unusable for further attacks.

\section*{Appendix B:  Vulnerabilities in the network}
TABLE \ref{tab:vul} represents the host/node in the network, its IP address. W-1, W-2, $\ldots$ F-3 represent various services running on each host, and their vulnerabilities along with the CVEs associated with each host.

\input{Files/tables/cve-table}

%% file: Files/tables/cve-table.tex
\begin{table}[!h]
\caption{Vulnerabilities (CVE ID) on each server.}
\label{tab:vul}
\begin{footnotesize}
\centering
\resizebox{\columnwidth}{!}{%
\begin{tabular}{|c|l|l|}
\hline
\textbf{Host/Node} &
  \textbf{Services: Vulnerabilities} &
  \textbf{CVE\#} \\ \hline
\multirow{2}{*}{\begin{tabular}[c]{@{}l@{}}Web Server\\ (192.168.1.1)\end{tabular}} &
  W-1: Remote code execution in HTTP.sys &
  CVE-2015-1635 \\ \cline{2-3} 
 &
  W-2: Remote code execution &
  CVE-2020-5902 \\ \hline
\multirow{2}{*}{\begin{tabular}[c]{@{}l@{}}DNS Server \\ (192.168.2.1)\end{tabular}} &
  D-1: Heap based buffer-overflow &
  CVE-2019-0211 \\ \cline{2-3} 
 &
  D-2: Remote code execution &
  CVE-2020-1350 \\ \hline
\multirow{2}{*}{\begin{tabular}[c]{@{}l@{}}Mail Server \\ (192.168.3.1)\end{tabular}} &
  M-1: Arbitrary code execution &
  CVE-2004-0840 \\ \cline{2-3} 
 &
  \begin{tabular}[c]{@{}l@{}}M-2: Sensitive information disclosure via \\ malformed input and invalid session ID\end{tabular} &
  CVE-2008-3060 \\ \hline
\multirow{3}{*}{\begin{tabular}[c]{@{}l@{}}Database \\ Server \\ (192.168.4.1)\end{tabular}} &
  S-1: Code execution with superuser permissions &
  CVE-2018-1058 \\ \cline{2-3} 
 &
  S-2: SSRF via malicious dockerd server &
  CVE-2020-13295 \\ \cline{2-3} 
 &
  S-3: Improper validation of cstrings &
  CVE-2021-20329 \\ \hline
\multirow{2}{*}{\begin{tabular}[c]{@{}l@{}}Admin \\ Server \\ (192.168.5.1)\end{tabular}} &
  A-1: Stack based buffer-overflow &
  CVE-2009-0241 \\ \cline{2-3} 
 &
  \begin{tabular}[c]{@{}l@{}}A-2: Arbitrary code execution via scoreboard \\ manipulation\end{tabular} &
  CVE-2022-37835 \\ \hline
\multirow{3}{*}{\begin{tabular}[c]{@{}l@{}}FTP \\ Server \\ (192.168.6.1)\end{tabular}} &
  F-1: Unrestricted file upload &
  CVE-2013-4465 \\ \cline{2-3} 
 &
  F-2: Out-of-bounds (OOB) read in function &
  CVE-2020-9365 \\ \cline{2-3} 
 &
  F-3: Windows FTP Server DoS &
  CVE-2018-8206 \\ \hline
\end{tabular}%
}
\end{footnotesize}
\end{table}